\def\llncs{1}
\def\fullpage{0}
\def\anonymous{0}
\def\authnote{1}
\def\notxfont{0}

\ifnum\llncs=1
	\documentclass[envcountsect,a4paper,runningheads,11pt]{llncs}
\else
	\documentclass[letterpaper,hmargin=1.05in,vmargin=1.05in]{article}
			\ifnum\fullpage=1
		\usepackage{fullpage}
		\fi
\fi

\usepackage[%
  colorlinks=true,
  citecolor=blue,
  pagebackref=true
]{hyperref}

\usepackage{amsmath, amsfonts, amssymb, mathtools,amscd}

\usepackage{amsthm}

\usepackage{lmodern}
\usepackage[T1]{fontenc}
\usepackage[utf8]{inputenc}

\usepackage{arydshln} 
\usepackage{url}
\usepackage{ifthen}
\usepackage{bm}
\usepackage{multirow}
\usepackage[dvips]{graphicx}
\usepackage[usenames]{color}
\usepackage{xcolor,colortbl} 
\usepackage{threeparttable}
\usepackage{comment}
\usepackage{paralist,verbatim}
\usepackage{cases}
\usepackage{booktabs}
\usepackage{braket}
\usepackage{cancel} 
\usepackage{ascmac} 
\usepackage{framed}
\usepackage{authblk}
\usepackage{pifont}
\usepackage{qcircuit}
\definecolor{darkblue}{rgb}{0,0,0.6}
\definecolor{darkgreen}{rgb}{0,0.5,0}
\definecolor{maroon}{rgb}{0.5,0.1,0.1}
\definecolor{dpurple}{rgb}{0.2,0,0.65}

\usepackage[capitalise,noabbrev]{cleveref}
\usepackage[absolute]{textpos}
\usepackage[final]{microtype}
\usepackage[absolute]{textpos}
\usepackage{everypage}
\DeclareMathAlphabet{\mathpzc}{OT1}{pzc}{m}{it}

\newtheoremstyle{thicktheorem}%
{\topsep}
{\topsep}
{\itshape}{}%
{\bfseries}%
{.}
{ }%
{\thmname{#1}\thmnumber{ #2}%
		\thmnote{ (#3)}%
}

\newtheoremstyle{remark}
{\topsep}
{\topsep}
	{}
	{}
	{}
	{.}
	{ }
	{\textit{\thmname{#1}}\thmnumber{ #2}
			\thmnote{ (#3)}%
	}

\ifnum\llncs=0
	\theoremstyle{thicktheorem}
	\newtheorem{theorem}{Theorem}[section]

	\newtheorem{definition}[theorem]{Definition}

	\theoremstyle{remark}
	
	\newtheorem{remark}[theorem]{Remark}

\else
\fi

\Crefname{MyClaim}{Claim}{Claims}

	\crefname{theorem}{Theorem}{Theorems}
	\crefname{assumption}{Assumption}{Assumptions}
	\crefname{construction}{Construction}{Constructions}
	\crefname{corollary}{Corollary}{Corollaries}
	\crefname{conjecture}{Conjecture}{Conjectures}
	\crefname{definition}{Definition}{Definitions}
	\crefname{exmaple}{Example}{Examples}
	\crefname{experiment}{Experiment}{Experiments}
	\crefname{counterexample}{Counterexample}{Counterexamples}
	\crefname{lemma}{Lemma}{Lemmata}
	\crefname{observation}{Observation}{Observations}
	\crefname{proposition}{Proposition}{Propositions}
	\crefname{remark}{Remark}{Remarks}
	\crefname{claim}{Claim}{Claims}
	\crefname{fact}{Fact}{Facts}
	\crefname{note}{Note}{Notes}

\ifnum\llncs=1
 \crefname{appendix}{App.}{Appendices}
 \crefname{section}{Sec.}{Sections}
\else
\fi

\ifnum\llncs=1
\renewcommand*{\backref}[1]{}
\else
	\renewcommand*{\backref}[1]{(Cited on page~#1.)}
	\ifnum\notxfont=1
	\else
		\usepackage{newtxtext}
	\fi
\fi
\ifnum\authnote=0  
\newcommand{\mor}[1]{}
\newcommand{\minki}[1]{}
\newcommand{\takashi}[1]{}

\else
\newcommand{\mor}[1]{$\ll$\textsf{\color{red} Tomoyuki: { #1}}$\gg$}
\newcommand{\takashi}[1]{$\ll$\textsf{\color{orange} Takashi: { #1}}$\gg$}
\newcommand{\minki}[1]{$\ll$\textsf{\color{darkgreen} Minki: { #1}}$\gg$}

\fi

\newcommand{\NP}{\mathbf{NP}}

\newcommand{\cA}{\mathcal{A}}

\newcommand{\cC}{\mathcal{C}}

\newcommand{\cK}{\mathcal{K}}

\newcommand{\cP}{\mathcal{P}}

\newcommand{\cV}{\mathcal{V}}

\newcommand{\cX}{\mathcal{X}}
\newcommand{\cY}{\mathcal{Y}}

\def\makeuppercase#1{
\expandafter\newcommand\csname tl#1\endcsname{\widetilde{#1}}
}

\def\makelowercase#1{
\expandafter\newcommand\csname tl#1\endcsname{\widetilde{#1}}
}

\newcommand{\secp}{\lambda}

\newcommand{\A}{\entity{A}}
\newcommand{\B}{\entity{B}}

\newcommand*{\sk}{\keys{sk}}
\newcommand*{\pk}{\keys{pk}}

\newcommand*{\td}{\keys{td}}

\newcommand*{\keys}[1]{\mathsf{#1}}

\newcommand*{\algo}[1]{\ensuremath{\mathsf{#1}}}

\newcommand*{\entity}[1]{\mathcal{#1}}

\newenvironment{boxfig}[2]{\begin{figure}[#1]\fbox{\begin{minipage}{0.97\linewidth}
                        \vspace{0.2em}
                        \makebox[0.025\linewidth]{}
                        \begin{minipage}{0.95\linewidth}
            {{
                        #2 }}
                        \end{minipage}
                        \vspace{0.2em}
                        \end{minipage}}}{\end{figure}}

\newcommand{\bit}{\{0,1\}}

\newcommand{\Gen}{\algo{Gen}}

\newcommand{\Inv}{\algo{Inv}}
\newcommand{\Ver}{\algo{Ver}}
\newcommand{\Comm}{\algo{Comm}}
\newcommand{\comm}{\algo{comm}}
\newcommand{\st}{\algo{st}}
\newcommand{\Ans}{\algo{Ans}}

\newcommand{\PoQ}{\algo{PoQ}}

\newcommand{\TD}{\algo{TD}}

\newcommand{\Eval}{\algo{Eval}}

\newcommand{\negl}{{\mathsf{negl}}}

\newcommand{\poly}{{\mathrm{poly}}}

\makeatletter
\DeclareRobustCommand
  \myvdots{\vbox{\baselineskip4\p@ \lineskiplimit\z@
    \hbox{.}\hbox{.}\hbox{.}}}
\makeatother

\title{Proofs of Quantumness from Trapdoor Permutations}

\ifnum\anonymous=1
\author{\empty}\institute{\empty}
\else
\ifnum\llncs=1
\author{
	Tomoyuki Morimae\inst{1} \and Takashi Yamakawa\inst{1,2}
}
\institute{
	Yukawa Institute for Theoretical Physics, Kyoto University, Kyoto, Japan \and NTT Social Informatics Laboratories, Tokyo, Japan
}
\else
\author[1]{Tomoyuki Morimae}
\author[1,2]{\hskip 1em Takashi Yamakawa}
\affil[1]{{\small Yukawa Institute for Theoretical Physics, Kyoto University, Kyoto, Japan}\authorcr{\small tomoyuki.morimae@yukawa.kyoto-u.ac.jp}}
\affil[2]{{\small NTT Social Informatics Laboratories, Tokyo, Japan}\authorcr{\small takashi.yamakawa.ga@hco.ntt.co.jp}}

\fi 
\fi

\date{}

\begin{document}

\maketitle

\begin{abstract}
Assume that Alice can do only classical probabilistic polynomial-time computing
while Bob can do quantum polynomial-time computing. Alice and Bob communicate over only
classical channels, and finally Bob gets a state $|x_0\rangle+|x_1\rangle$ with
some bit strings $x_0$ and $x_1$.
Is it possible that Alice can know $\{x_0,x_1\}$ but Bob cannot?
Such a task, called {\it remote state preparations}, is indeed possible under some 
complexity assumptions, 
and 
is bases of many quantum cryptographic primitives such as
proofs of quantumness, (classical-client) blind quantum computing, (classical) verifications of quantum computing,
and quantum money.
A typical technique to realize remote state preparations is to use 2-to-1 trapdoor collision resistant hash functions:
Alice sends a 2-to-1 trapdoor collision resistant hash function $f$ to Bob, and Bob evaluates it coherently,
i.e., Bob generates $\sum_x|x\rangle|f(x)\rangle$.
Bob measures the second register to get the measurement result $y$,
and sends $y$ to Alice. Bob's post-measurement state is $|x_0\rangle+|x_1\rangle$,
where $f(x_0)=f(x_1)=y$. With the trapdoor, Alice can learn $\{x_0,x_1\}$ from $y$,
but due to the collision resistance, Bob cannot. This
Alice's advantage can be leveraged to realize the quantum cryptographic primitives listed above.
It seems that the collision resistance is essential here.
In this paper, surprisingly, we show that the collision resistance is not necessary for a restricted case:
we show that (non-verifiable) remote state preparations of $|x_0\rangle+|x_1\rangle$
secure against {\it classical} probabilistic polynomial-time Bob can be constructed from classically-secure
(full-domain) trapdoor permutations.
Trapdoor permutations are not likely to imply the collision resistance, because 
black-box reductions from collision-resistant hash functions to trapdoor permutations are known to be impossible.
As an application of our result, we construct proofs of quantumness from
classically-secure (full-domain) trapdoor permutations. 
\end{abstract}

\section{Introduction}
Let us consider a two-party interactive protocol between Alice and Bob.
Alice can do only classical probabilistic polynomial-time computing
while Bob can do quantum polynomial-time computing. Alice and Bob communicate over only
classical channels. After the interaction, Alice finally outputs a pair 
$\{x_0,x_1\}$
of $n$-bit strings
$x_0,x_1\in\bit^n$. 
If Bob behaves honestly, he finally outputs the $n$-qubit state $|x_0\rangle+|x_1\rangle$.\footnote{For simplicity, in this paper, we often omit the normalization factors of quantum states.}
On the other hand, no malicious Bob can learn $\{x_0,x_1\}$. 
Such a task, called {\it remote state preparations}~\cite{DK16,AC:CCKW19,FOCS:GheVid19}\footnote{Generally speaking, remote state preparations are the task that a classical Alice
delegates preparations of quantum states to quantum Bob over a classical channel in such a way that Bob cannot learn which states he is generating. \cite{AC:CCKW19} considered remote state preparations of
random single-qubit states for the applications to classical-client blind quantum computing. 
In this paper, on the other hand, we focus on remote state preparations of an equal-weight superposition $|x_0\rangle+|x_1\rangle$ of two $n$-qubit computational basis states. Moreover, note that there are also {\it verifiable} remote state preparations~\cite{FOCS:GheVid19}
where Alice can check whether Bob has generated correct states or not. Verifiability is not necessary for some applications such as classical-client blind quantum computing,
but seems to be necessary for some applications such as classical verifications of quantum computing. In this paper, we do not consider the verifiability.},
is indeed possible under some complexity assumptions,
and is bases of many quantum cryptographic primitives such as
proofs of quantumness~\cite{JACM:BCMVV21}, 
(classical-client) blind quantum computing~\cite{FOCS:BroFitKas09,AC:CCKW19,AC:BCCKLMW20}, 
(classical) verifications of quantum computing~\cite{FOCS:Mahadev18a,FOCS:GheVid19},
and quantum money~\cite{CoRR:RadSat19}.
In fact, if Alice can generate a quantum state $|x_0\rangle+|x_1\rangle$ and send it
to Bob over a quantum channel, Alice can enjoy several advantages over Bob.
For example, 
Bob cannot know the complete classical description of the state, 
he cannot clone the state, and he has to disturb the state if he measures it, etc.
Such an inequivalence between Alice and Bob is clearly useful for quantum cryptography.
Remote state preparations somehow ``simulate'' such situations, and can replace quantum channels
with classical channels for some applications.

A typical technique to realize remote state preparations is to use 2-to-1 trapdoor collision resistant hash functions~\cite{JACM:BCMVV21}\footnote{A function $f:\cX\to\cY$ is 2-to-1 if
$|\{x\in \cX~|~f(x)=y\}|=2$ for all $y\in\cY$. A function $f$ is called a trapdoor collision resistant hash function if given $f$ it is hard to find
$x$ and $x'$ such that $f(x)=f(x')$, but it becomes easy if a trapdoor is available.}: 
Alice sends a 2-to-1 trapdoor collision resistant hash function $f$ to Bob, and Bob evaluates it coherently,
i.e., Bob generates $\sum_x|x\rangle|f(x)\rangle$.
Bob measures the second register to get the measurement result $y$,
and sends $y$ to Alice. Bob's post-measurement state is $|x_0\rangle+|x_1\rangle$,
where $f(x_0)=f(x_1)=y$. With the trapdoor, Alice can learn $\{x_0,x_1\}$ from $y$,
but due to the collision resistance, Bob cannot. This
Alice's advantage can be leveraged to realize the quantum cryptographic primitives listed above.

\subsection{Our Results}
The collision resistance seems to be essential for remote state preparations (and many other quantum cryptographic primitives over classical channels.)
In this paper, surprisingly, we show that the collision resistance is not necessary for a restricted case.
We show the following result.
\begin{theorem}
\label{theorem:result1}
(Non-verifiable) remote state preparations of $|x_0\rangle+|x_1\rangle$
secure against {\bf classical} probabilistic polynomial-time Bob can be constructed from classically-secure (full-domain) trapdoor permutations.
\end{theorem}
Here, non-verifiable remote state preparations are remote state preparations that are blind but not verifiable,
i.e., no malicious Bob can learn 
which states he is generating, but Alice cannot verify whether Bob has generated correct states or not.
(A formal definition is given in \cref{def:RSP}.)
A classically-secure trapdoor permutation is a permutation $f:\cX\to\cX$ such that inverting it is hard for classical
probabilistic polynomial-time adversaries, but
it is easy if a trapdoor is available. Full-domain means that $\cX=\bit^n$.
(The formal definition is given in \cref{def:TDP}.)
A proof of \cref{theorem:result1} is given in \cref{sec:proof1}.

Trapdoor permutations are not likely to imply the collision resistance, because 
black-box reductions from collision-resistant hash functions to trapdoor permutations are known to be impossible~\cite{HHRS15,AC:HosYam20}.
Classically-secure full-domain trapdoor permutations can be instantiated
with the hardness of factoring~\cite{JACM:BM92,JC:GolRot13}.

We emphasize that our remote state preparations in \cref{theorem:result1} are proven to be 
secure against only {\it classical} Bob,
which unfortunately restricts the applications of our result.
We do not know how to achieve the quantum security. This is an important open problem.
(For more discussion, see \cref{sec:open}.)

As an application of \cref{theorem:result1}, we construct proofs of quantumness.
\begin{theorem}
\label{theorem:result2}
Proofs of quantumness can be constructed from classically-secure (full-domain) trapdoor permutations.
\end{theorem}
Its proof is given in \cref{sec:proof2}.
Proofs of quantumness are two-party protocols between a probabilistic polynomial-time verifier and a prover. 
A quantum polynomial-time prover can make the verifier accept with high probability, but
no probabilistic polynomial-time prover is accepted by the verifier except for a negligible probability.
(For a formal definition of proofs of quantumness, see \cref{def:PoQ}.)
The first construction of proofs of quantumness \cite{JACM:BCMVV21} is based on 
trapdoor injective claw-free functions 
with the special property called adaptive-hardcore-bit property.
Here, an injective claw-free means that given a pair $(f_0,f_1)$ of injective functions, it is hard to find
$(x_0,x_1)$ such that $f_0(x_0)=f_1(x_1)$.
The adaptive-hardcore-bit property
roughly means that given $(f_0,f_1)$,
it is hard to
find one $x_b$ of a claw $(x_0,x_1)$
and $d$ such that $d\cdot(x_0\oplus x_1)=0$ at the same time.
It is easy to see that injective claw-free functions imply the collision resistance.
\footnote{Let us define a function $g$ by $g(0x)\coloneqq f_0(x)$ and $g(1x)\coloneqq f_1(x)$. Assume that
a collision $(\alpha,\beta)$ of $g$ is easily found, i.e., $g(\alpha)=g(\beta)$. Then,
due to the injectivity of $f_0$ and $f_1$, the first bit of $\alpha$ and $\beta$ should be different.
Without loss of generality, assume that $\alpha=0x_0$ and $\beta=1x_1$. Then,
$(x_0,x_1)$ is a claw of $(f_0,f_1)$.}
A recent paper \cite{NatPhys:KMCVY22} has improved the result of \cite{JACM:BCMVV21} by removing
the necessity of the adaptive-hardcore-bit property.
However, it still uses
2-to-1 trapdoor collision resistant hash functions.
Our \cref{theorem:result2} removes the necessity of the collision resistance of \cite{NatPhys:KMCVY22}.

\subsection{Technical Overview}
Our construction of remote state preparations from trapdoor permutations (\cref{theorem:result1})
is based on the statistically-hiding and computationally-binding commitment scheme from one-way permutations~\cite{C:NOVY92},
which is explained as follows.
Let $f:\bit^n\to\bit^n$ be a one-way permutation.
\begin{enumerate}
    \item The sender of the commitment scheme chooses $x\gets\bit^n$, and computes $y\coloneqq f(x)$.
    \item The receiver of the commitment scheme chooses $h_j\gets 0^{j-1}1\bit^{n-j}$ for each $j=1,2,...,n-1$.
    \item The receiver and the sender repeat the following procedure for $j=1,2,...,n-1$:
    \begin{enumerate}
    \item The receiver sends $h_j$ to the sender.
    \item The sender returns the value $c_j\coloneqq h_j\cdot y$ to the receiver.\footnote{For two bit strings $a,b\in\bit^n$, $a\cdot b$ is the bitwise inner product, i.e., $a\cdot b\coloneqq\bigoplus_{j=1}^n a_jb_j$.}
    \end{enumerate}
\item
The receiver and the sender finally obtain the system of linear equations $\{h_j\cdot y=c_j\}_{j=1}^{n-1}$ that has two
solutions $y_0,y_1\in\bit^n$, where $y_0$ is the lexicographically smaller one.
Let $c\in\bit$ be such that $y_c=y$. 
The sender sends $b\oplus c$ to the receiver as the commitment of the bit $b\in\bit$.
\item
The opening for the commitment is $x$ and $b$.
\end{enumerate}
It is shown in \cite{C:NOVY92} that if a probabilistic polynomial-time sender can find both $x_0$ and $x_1$ such that $f(x_0)=y_0$ and $f(x_1)=y_1$
with a non-negligible probability, then
a probabilistic polynomial-time adversary that breaks the security of the one-way permutation $f$ can be constructed, which shows the (classical) computational binding of the scheme.

What happens if this NOVY's interactive hashing is run coherently? Namely, let us consider the following ``quantum version'' of the NOVY's interactive hashing:
\begin{enumerate}
\item A quantum polynomial-time Bob prepares $\sum_{x\in\bit^n}|x\rangle$. 
\item A probabilistic polynomial-time Alice chooses $h_j\gets 0^{j-1}1\bit^{n-j}$ for each $j=1,2,...,n-1$.
\item Alice sends $h_1$ to Bob.
\item Bob generates 
\begin{eqnarray*}
\sum_{x\in \bit^n}|x\rangle|h_1\cdot f(x)\rangle, 
\end{eqnarray*}
measures the second register to get the measurement result $c_1\in\bit$,
and sends $c_1$ to Alice.
The post-measurement state is
\begin{eqnarray*}
\sum_{x\in \bit^n:h_1\cdot f(x)=c_1}|x\rangle. 
\end{eqnarray*}
\item Alice sends $h_2$ to Bob.
\item Bob generates 
\begin{eqnarray*}
\sum_{x\in \bit^n:h_1\cdot f(x)=c_1}|x\rangle|h_2\cdot f(x)\rangle, 
\end{eqnarray*}
measures the second register to get the measurement result $c_2\in\bit$,
and sends $c_2$ to Alice.
The post-measurement state is
\begin{eqnarray*}
\sum_{x\in \bit^n:h_1\cdot f(x)=c_1, h_2\cdot f(x)=c_2}|x\rangle. 
\end{eqnarray*}
\item
If they repeat the above procedure for $j=3,4,...,n-1$,
Bob finally possesses the state $|x_0\rangle+|x_1\rangle$, where $f(x_0)=y_0$, $f(x_1)=y_1$, and
$(y_0,y_1)$ is the two solutions of 
$\{h_j\cdot y=c_j\}_{j=1}^{n-1}$.
\end{enumerate}
By the (classical) computational-binding of the NOVY's commitment scheme, no probabilistic polynomial-time Bob can learn both $x_0$ and $x_1$ at the same time with non-negligible probability.
If $f$ is a trapdoor permutation, Alice can compute both $x_0$ and $x_1$,
and it is clear that the existence of the trapdoor does not degrade the security against Bob.

In this way, we can construct remote state preparations of $|x_0\rangle+|x_1\rangle$
secure against classical Bob
from trapdoor permutations.
Proofs of quantumness can be constructed from it based on a similar idea of \cite{NatPhys:KMCVY22}, 
which is a proof of \cref{theorem:result2}:
The verifier and the prover first run remote state preparations of $|x_0\rangle+|x_1\rangle$.
Then, with probability 1/2, the verifier asks the prover to measure the state in the computational basis.
If the prover returns a correct $x_0$ or $x_1$, the verifier accepts.
With probability 1/2,
the verifier chooses $r\gets\bit^n$ and sends it to the prover.
The prover generates the state
$|r\cdot x_0\rangle|x_0\rangle+|r\cdot x_1\rangle|x_1\rangle$, and measures the second register in the
Hadamard basis. If the measurement result is $d\in\bit^n$, the post-measurement state is
$|r\cdot x_0\rangle+(-1)^{d\cdot(x_0\oplus x_1)}|r\cdot x_1\rangle$,
i.e., one of the BB84 states. Then, runing the CHSH game on it leads to proofs of quantumness:
the quantum polynomial-time prover can output the correct measurement result with high probability, but no probabilistic polynomial-time prover can output the correct measurement result
except for a small probability.

\subsection{Open Problems}
\label{sec:open}
Our novel idea of running the NOVY's interactive hashing coherently will have many other interesting applications.
Let us summarize several open problems.

\paragraph{\bf Weakening assumptions.}
Our constructions of remote state preparations and proofs of quantumness are based on
the full-domain trapdoor permutations. One important open problem is whether the assumption can be weakened or not.

First, can we remove the full-domain property?
Though the factoring-based constructions can be made full-domain~\cite{JACM:BM92,JC:GolRot13}, it is not true in general. 
For example, the construction of trapdoor permutations based on indistinguishability obfuscation (iO)~\cite{TCC:BitPanWic16} is not full-domain. 
In many applications of trapdoor permutations like oblivious transfers~\cite{EGL85} and non-interactive zero-knowledge~\cite{FLS99}, the full-domain property can be weakened to a property called the \emph{doubly enhanced} property~\cite{JC:GolRot13}.
Can we replace full-domain trapdoor permutations in our constructions with (non-full-domain) doubly enhanced trapdoor permutations? 
If this is possible, then we would obtain proofs of quantumness from iO and one-way functions since the iO-based trapdoor permutation~\cite{TCC:BitPanWic16} satisfies the doubly enhanced property.
(Or, it is an interesting open problem whether remote state preparations or proofs of quantumness can be directly constructed from iO (plus one-way functions).)

Second, it is known that the NOVY's commitment scheme with one-way permutations can be improved to that with regular one-way functions~\cite{EC:HHKKMS05} or even any one-way functions~\cite{HNORV09}.\footnote{We say that a one-way functions is regular if the preimage sizes are equal for all images.}
Can we construct remote state preparations or proofs of quantumness from trapdoor functions, not permutations?
Known instantiations of trapdoor permutations are only factoring-based ones (for full-domain cases) and iO-based ones (for non-full-domain cases).
However, trapdoor functions have many instantiations from
Diffie-Hellman assumptions~\cite{STOC:PeiWat08,C:GarHaj18},
learning with errors~\cite{STOC:PeiWat08,STOC:GenPeiVai08},
NTRU~\cite{ANTS:HofPipSil98}, 
and coding theory \cite{McEliece,Niederreiter}, etc.
A potential approach is to coherently run the variant of the NOVY's commitment based on regular one-way functions~\cite{EC:HHKKMS05} since trapdoor functions are automatically regular. However, in such a construction, the state which the honest Bob gets after the interaction is not a superposition of two computational-basis states, $|x_0\rangle+|x_1\rangle$, 
but that of many computational-basis states, i.e., $\sum_{x\in S}|x\rangle$ with $|S|$ being polynomially or even exponentially large.
We do not know how to construct proofs of quantumness from such a superposition of many computational-basis states.

Finally, can we remove even the trapdoor? Is it possible to construct remote state preparations or proofs of quantumness
from one-way functions? 
A recent work \cite{cryptoeprint:2022/434} constructs proofs of quantumness in the quantum random oracle model. This demonstrates that trapdoors are not inherent for proofs of quantumness. However, a standard-model instantiation of their protocol based on standard assumptions is likely to require a completely new idea.

\paragraph{\bf Other applications.}
The other important open problem is whether
we can construct other quantum cryptographic primitives, such as (classical-client) blind quantum computing
and (classical) verifications of quantum computing from trapdoor permutations (or even trapdoor/one-way functions). 
For that goal, we have to show the {\it quantum} security of the NOVY's scheme.
(Remember that our remote state preparations are secure only against {\it classical} Bob, because we do not show the security of the NOVY's scheme
against quantum adversaries.)
The security proof of \cite{C:NOVY92} makes heavy use of the rewinding technique, and therefore
showing the quantum security of NOVY's scheme seems to be a challenging open problem, which is of independent interest.

\subsection{Related Works}
\paragraph{\bf Remote state preparations.}
The first construction of (non-verifiable) remote state preparations of single-qubit states so called ``QFactory''~\cite{AC:CCKW19} used 
certain 2-to-1 trapdoor collision resistant hash functions with some homomorphic predicates, which can be constructed from the LWE assumption.
\cite{FOCS:GheVid19,cryptoeprint:2022/432,GMP22}
constructed verifiable remote state preparations of single-qubit states
that use the (noisy) trapdoor injective claw-free functions of \cite{FOCS:BCMVV18}.

\paragraph{\bf Proofs of quantumness.}
A simple way of achieving ``proofs of quantumness'' is to ask the prover to solve (non-interactive) problems in $\NP$ that
can be solved in quantum polynomial-time but are believed to be hard for probabilistic polynomial-time, such as
factoring~\cite{FOCS:Shor94}, Pell's equation~\cite{STOC:Hallgren02}, and matrix group membership~\cite{STOC:BabBeaSer09}, etc.
In this paper, however, we do not consider such approaches.

The original construction~\cite{JACM:BCMVV21} of proofs of quantumness required the adaptive-hardcore-bit property.
\cite{TQC:BKVV20,NatPhys:KMCVY22} removed the necessity of the adaptive-hardcore-bit property,
but \cite{TQC:BKVV20} used 2-to-1 trapdoor injective claw-free functions and random oracles that can be queried coherently, 
and \cite{NatPhys:KMCVY22} used 2-to-1 trapdoor collision resistant hash functions.

Publicly-verifiable proofs of quantumness were also studied.
One-shot signatures~\cite{STOC:AGKZ20} imply proofs of quantumness, but
the known construction of one-shot signatures is based on one-shot chameleon hash functions, which 
satisfy the collision resistance. Moreover, the known construction assumes classical oracles that can be
queried coherently.
\cite{cryptoeprint:2022/434} constructed publicly-verifiable non-interactive proofs of quantumness
with random oracles. Note that random oracles are collision-resistant.

Recently, \cite{cryptoeprint:2022/400} showed a general compiler to transform non-local games to proofs of quantumness via quantum homomorphic encryptions (QHE).
Their assumption is only the existence of QHE for certain class of quantum operations (such as controlled-Hadamard gates), which can be instantiated with LWE~\cite{FOCS:Mahadev18b}. 
Although homomorphic encryptions generally imply the collision resistance~\cite{TCC:IshKusOst05}, it is not known whether the restricted QHE used in \cite{cryptoeprint:2022/400} implies the collision resistance.

The idea of \cite{cryptoeprint:2022/400} is that the quantum prover generates a bipartite state, and the classical verifier 
remotely measures one of the registers via QHE so that the prover gets an unknown state of the other register.
It therefore can be also considered as (non-verifiable) remote state preparations via QHE.
Remote state preparations via QHE were also introduced in \cite{cryptoeprint:2021/1427,cryptoeprint:2022/228} in the contexts of quantum money and quantum tokenized signatures.

\section{Preliminaries}
We use the standard notations of quantum computing and cryptography.
We use $\secp$ as the security parameter.
$x\gets A$ means that an element $x$ is sampled uniformly at random from the set $A$.
$\negl$ is a negligible function, and $\poly$ is a polynomial.
For an algorithm $A$, $y\gets A(x)$ means that the algorithm $A$ outputs $y$ on input $x$.
For two bit strings $a,b\in\bit^n$, $a\cdot b$ is the bitwise inner product, i.e., $a\cdot b\coloneqq\bigoplus_{j=1}^n a_jb_j$.

Non-verifiable remote state preparations of $|x_0\rangle+|x_1\rangle$ secure against probabilistic polynomial-time Bob are defined as follows.
\begin{definition}[Remote state preparations]
\label{def:RSP}
Non-verifiable remote state preparations of $|x_0\rangle+|x_1\rangle$ secure against probabilistic polynomial-time Bob
are two-party interactive protocols
between probabilistic polynomial-time Alice and quantum/probabilistic polynomial-time Bob over a classical channel that
satisfy the following
two conditions.

\paragraph{\bf Perfect correctness:}
If quantum polynomial-time Bob behaves honestly, 
Alice outputs a pair $\{x_0,x_1\}$ of two $n$-bit strings $x_0,x_1\in\bit^n$ and
Bob outputs the $n$-qubit state $|x_0\rangle+|x_1\rangle$ with probability 1.

\paragraph{\bf Classical security (blindness):}
For any probabilistic polynomial-time malicious Bob that outputs a pair $\{\alpha,\beta\}$ of two $n$-bit strings $\alpha,\beta\in\bit^n$, 
\begin{eqnarray*}
\Pr[\{x_0,x_1\}=\{\alpha,\beta\}:\{x_0,x_1\}\gets\mbox{Alice},\{\alpha,\beta\}\gets\mbox{Bob}]\le\negl(\secp).
\end{eqnarray*}
\end{definition}

\begin{remark}
Our definition is different from previous ones \cite{AC:CCKW19,FOCS:GheVid19} in the following two points.
First, they are interested in remotely generating single-qubit states while we consider remote generations of $|x_0\rangle+|x_1\rangle$. 
Second, \cite{AC:CCKW19,FOCS:GheVid19} consider the security against quantum Bob, while we consider the one against only
{\it classical} Bob. It is an important open problem whether we can show the quantum security.
\end{remark}

\begin{remark}
Remote state preparations can have the {\it verifiability}~\cite{FOCS:GheVid19}, which roughly means that Alice can
check whether Bob has generated correct states or not. Some applications, such as classical-client blind quantum computing,
do not require the verifiability, but it seems that verifiability is necessary for some applications, such as classical verifications of
quantum computing. In this paper, we do not consider verifiability.
\end{remark}

\begin{remark}
We can consider non-perfect correctness, but in this paper we consider only perfect correctness, because our construction satisfies it.
Moreover, it is reasonable to assume in the definition that Alice sometimes outputs $\bot$.
However, for simplicity, we assume that Alice never outputs $\bot$. In fact, our construction satisfies it.
\end{remark}

Proofs of quantumness are defined as follows.
\begin{definition}[Proofs of Quantumness]
\label{def:PoQ}
Proofs of quantumness are two-party protocols between a probabilistic polynomial-time verifier $\cV$ and a quantum/probabilistic polynomial-time
prover $\cP$ over a classical channel
such that $\cV$ finally outputs $\top$ or $\bot$.
We require that the following two conditions, $\alpha$-correctness and $\beta$-soundness, are satisfied for
some $\alpha$ and $\beta$ such that $\alpha-\beta\ge\frac{1}{\poly(\secp)}$.

\paragraph{\bf $\alpha$-correctness:}
There exists a quantum polynomial-time prover $\cP$ such that
$\Pr\left[\cV\to\top\right] \ge \alpha$.

\paragraph {\bf $\beta$-soundness:}
For any probabilistic polynomial-time prover $\cP$, 
$\Pr\left[\cV\to\top\right]
\le \beta$.
\end{definition}

Classically-secure full-domain trapdoor permutations are defined as follows. 

\begin{definition}[Trapdoor Permutations]
\label{def:TDP}
A family $\{f_k\}_{k\in {\mathcal K}_\secp}$ of permutations is called a classically-secure full-domain trapdoor permutation family if
there is a tuple of algorithms $(\Gen,\Eval,\Inv)$ such that
\begin{itemize}
    \item 
    $\Gen(1^\secp)\to(\td,k):$ It is a probabilistic polynomial-time algorithm that, on input the security parameter $\secp$,
    outputs a trapdoor $\td$ and a key $k$.
    \item
    $\Eval(x,k)\to x':$ It is a classical polynomial-time deterministic algorithm that, on input $k$ and a bit string $x\in\{0,1\}^n$,
    outputs a bit string $x'\in\{0,1\}^n$.
    \item
    $\Inv(\td,x')\to x'':$ It is a classical polynomial-time deterministic algorithm that, on input $x'$ and $\td$, outputs a bit string $x''$.
\end{itemize}

We require the following two types of correctness and the security.

\paragraph{\bf Evaluation correctness:}
For any $x\in\{0,1\}^n$,
\begin{eqnarray*}
\Pr[x'=f_k(x):
x'\leftarrow \Eval(x,k),
(\td,k)\leftarrow \Gen(1^\secp)
]=1.
\end{eqnarray*}

\paragraph{\bf Inversion correctness:}
For any $x\in\bit^n$,
\begin{eqnarray*}
\Pr[x''=x:
x''\leftarrow \Inv(\td,x'),
x'\leftarrow \Eval(x,k),
(\td,k)\leftarrow \Gen(1^\secp)
]=1.
\end{eqnarray*}

\paragraph{\bf Classical security:} 
For any probabilistic polynomial-time adversary $\cA$,
\begin{eqnarray*}
\Pr[x''=x:
x''\leftarrow \cA(x',k),
x'\leftarrow \Eval(x,k),
x\leftarrow \{0,1\}^n,
(\td,k)\leftarrow \Gen(1^\secp)
]\le \negl(\secp).
\end{eqnarray*}
\end{definition}

We use the following result from \cite{C:NOVY92}.
(They show it for one-way permutations, not trapdoor permutations, but it is clear that the same result holds
for the trapdoor permutations.)
\begin{theorem}[\cite{C:NOVY92}]
\label{theorem:NOVY}
Let $\{f_k\}_{k\in{\mathcal K}_\secp}$ be a classically-secure full-domain trapdoor permutation family.
Let $(\Gen,\Eval,\Inv)$ be the associated tuple of algorithms.
Let us consider the following security game between a probabilistic polynomial-time challenger $\cC$ and a probabilistic polynomial-time adversary $\cA$.
\begin{enumerate}
    \item 
    $\cC$ runs $(\td,k)\leftarrow \Gen(1^\secp)$. 
    $\cC$ sends $k$ to $\cA$.
    \item
    $\cC$ chooses $h_j\gets0^{j-1}1\bit^{n-j}$ for each $j\in\{1,2,...,n-1\}$.
    \item
    $\cC$ and $\cA$ repeat the following for $j=1,2,...,n-1$:
    \begin{enumerate}
        \item $\cC$ sends $h_j$ to $\cA$.
        \item $\cA$ sends $c_j\in\bit$ to $\cC$.
    \end{enumerate}
    \item
    $\cA$ sends $\alpha,\beta\in\bit^n$ to $\cC$.
    \item
    There exist exactly two $y_0,y_1\in\bit^n$ such that $h_j\cdot y_b=c_j$ for all $b\in\bit$ and all $j\in\{1,2,...,n-1\}$.
    $\cC$ outputs $\top$ if $f_k(\alpha)=y_0$ and $f_k(\beta)=y_1$,
    or
    $f_k(\alpha)=y_1$ and $f_k(\beta)=y_0$.
    Otherwise, $\cC$ outputs $\bot$.
\end{enumerate}
Then, for any probabilistic polynomial-time adversary $\cA$,
$\Pr[\cC\to\top]\le\negl(\secp)$.
\end{theorem}

\section{Proof of \cref{theorem:result1}}
\label{sec:proof1}
In this section, we provide a proof of \cref{theorem:result1}.

\begin{proof}[Proof of Theorem~\ref{theorem:result1}]
Let $\{f_k\}_{k\in\cK_\secp}$ be a classically-secure full-domain trapdoor permutation family.
Let $(\Gen,\Eval,\Inv)$ be the associated tuple of algorithms.
From them, we construct non-verifiable remote state preparations of $|x_0\rangle+|x_1\rangle$ secure against probabilistic polynomial-time Bob as follows.
\begin{enumerate}
    \item 
    Alice runs $(\td,k)\leftarrow \Gen(1^\secp)$.
    Alice sends $k$ to Bob.
    \item
    Alice chooses $h_j\leftarrow 0^{j-1}1\{0,1\}^{n-j}$ for each $j\in\{1,2,...,n-1\}$.
    \item
    Alice and Bob repeat the following for $j=1,2,...,n-1$:
    \begin{enumerate}
    \item
    Alice sends $h_j$ to Bob.
    \item
    Bob possesses the state $\sum_{x\in X_{j-1}}|x\rangle$, where
    \begin{eqnarray*}
    X_j\coloneqq\Big\{x\in\bit^n~\Big|~\bigwedge_{i=1}^j(h_i\cdot f_k(x)=c_i)\Big\}.
    \end{eqnarray*}
    Bob generates
    \begin{eqnarray*}
    \sum_{x\in X_{j-1}}|x\rangle|h_j\cdot f_k(x)\rangle, 
    \end{eqnarray*}
    measures the second register to get the measurement result $c_j\in\bit$, and
    sends $c_j$ to Alice.
    The post-measurement state of the first register is
    \begin{eqnarray*}
    \sum_{x\in X_j}|x\rangle. 
    \end{eqnarray*}
   \end{enumerate} 
    \item
    Bob finally gets the state $|x_0\rangle+|x_1\rangle$,
    where there are exactly two bit strings $y_0,y_1\in\bit^n$ such that
    $h_j\cdot y_b=c_j$ for all $b\in\bit$ and all $j\in\{1,2,...,n-1\}$,
    and $f_k(x_0)=y_0$ and $f_k(x_1)=y_1$. 
    Bob outputs the state.
    \item
    From $\td$, $\{h_j\}_{j=1}^{n-1}$ and $\{c_j\}_{j=1}^{n-1}$,
    Alice computes $\{x_0,x_1\}$ and outputs it.
    \end{enumerate}
   The perfect correctness is clear. The classical security is obtained from \cref{theorem:NOVY}.
   \end{proof}

\section{Proof of \cref{theorem:result2}}  
\label{sec:proof2}

In this section, we show \cref{theorem:result2}.

   \begin{proof}[Proof of \cref{theorem:result2}]
Let us consider the following construction of proofs of quantumness, which is similar to that of \cite{NatPhys:KMCVY22}.
   \begin{enumerate} 
   \item \label{item:RSP}
   The verifier $\cV$ and the prover $\cP$ run non-verifiable remote state preparations of $|x_0\rangle+|x_1\rangle$ secure against probabilistic polynomial-time Bob whose existence is guaranteed from the existence of
   classically-secure full-domain trapdoor permutations due to \cref{theorem:result1}.
   $\cV$ gets a pair $\{x_0,x_1\}$ of $n$-bit strings $x_0,x_1\in\bit^n$. Honest $\cP$ gets the state $|x_0\rangle+|x_1\rangle$. 
    \item
    $\cV$ chooses $v_1\leftarrow \bit$. 
    $\cV$ chooses $r\leftarrow \bit^n$. 
    $\cV$ sends $v_1$ and $r$ to $\cP$. 
    \item \label{item:v1}
    \begin{itemize}
    \item
    If $v_1=0$: $\cP$ measures $|x_0\rangle+|x_1\rangle$ in the computational basis to get the measurement result $x\in\{0,1\}^n$.
    $\cP$ sends $x$ to $\cV$. If $x\in\{x_0,x_1\}$, $\cV$ outputs $\top$ and terminates the protocol.
    Otherwise, $\cV$ outputs $\bot$ and aborts.
    \item
    If $v_1=1$: $\cP$ changes $|x_0\rangle+|x_1\rangle$ into
    \begin{eqnarray*}
    |r\cdot x_0\rangle|x_0\rangle+|r\cdot x_1\rangle|x_1\rangle,
    \end{eqnarray*}
    measures its second register in the Hadamard basis to get the measurement result $d\in\{0,1\}^n$,
    and sends $d$ to $\cV$.
    The post-measurement state of the first register is
    \begin{eqnarray*}
    |r\cdot x_0\rangle+(-1)^{d\cdot(x_0\oplus x_1)}|r\cdot x_1\rangle.
    \end{eqnarray*}
    \end{itemize}
    \item
    $\cV$ chooses $v_2\leftarrow\bit$. $\cV$ sends $v_2$ to $\cP$.
    \item \label{item:v2}
    $\cP$ measures the state in the basis 
    \begin{eqnarray*}
\left\{\cos\frac{\pi}{8}|0\rangle+\sin\frac{\pi}{8}|1\rangle,
\sin\frac{\pi}{8}|0\rangle-\cos\frac{\pi}{8}|1\rangle\right\}
\end{eqnarray*}
if $v_2=0$, and
in the basis 
\begin{eqnarray*}
\left\{\cos\frac{\pi}{8}|0\rangle-\sin\frac{\pi}{8}|1\rangle,
\sin\frac{\pi}{8}|0\rangle+\cos\frac{\pi}{8}|1\rangle\right\}
\end{eqnarray*}
if $v_2=1$.
    Let $\eta\in\bit$ be the measurement result. (For the measurement in the basis $\{|\phi\rangle,|\phi^\perp\rangle\}$, the result 0 corresponds to $|\phi\rangle$ and the result 1 corresponds to $|\phi^\perp\rangle$.)
    $\cP$ sends $\eta$ to $\cV$.
    \item
    $\cV$ outputs $\top$ if 
    \begin{eqnarray*}
    (r\cdot x_0=r\cdot x_1) \wedge
    (\eta=r\cdot x_0),
    \end{eqnarray*}
    or
    \begin{eqnarray*}
    (r\cdot x_0\neq r\cdot x_1) \wedge
    (\eta=v_2\oplus d\cdot(x_0\oplus x_1)).
    \end{eqnarray*}
    Otherwise, it outputs $\bot$. 
\end{enumerate}

For the correctness, 
a straightforward calculation similarly to \cite{NatPhys:KMCVY22} shows that
\begin{eqnarray*}
\Pr[\cV\to\top]=\frac{1}{2}+\frac{1}{2}\cos^2\frac{\pi}{8}\simeq0.925 .
\end{eqnarray*}

For the soundness, by almost the same argument as that in \cite{NatPhys:KMCVY22}, we can show that for any probabilistic polynomial-time cheating prover, we have 
\begin{eqnarray*} 
\Pr[\cV\to\top]\leq \frac{7}{8}+\negl(\secp).
\end{eqnarray*}
Note that $\frac{7}{8}=0.875<0.925$.
We include the full proof in Appendix~\ref{app:soundness} for completeness. 


\end{proof}

{\bf Acknowledgements.}
TM is supported by
JST Moonshot R\verb|&|D JPMJMS2061-5-1-1, 
JST FOREST, 
MEXT QLEAP, 
the Grant-in-Aid for Scientific Research (B) No.JP19H04066, 
the Grant-in Aid for Transformative Research Areas (A) 21H05183,
and 
the Grant-in-Aid for Scientific Research (A) No.22H00522.

\bibliographystyle{alpha} 
\bibliography{abbrev3,crypto,reference}

\appendix 
\section{Proof of Soundness}\label{app:soundness}
We give the omitted proof of the soundness of the proof of quantumness protocol given in Section~\ref{sec:proof2}. Note that this is almost identical to that in \cite{NatPhys:KMCVY22}. 

Our goal is to prove that 
for any probabilistic polynomial-time cheating prover,
\begin{eqnarray} \label{eq:goal}
\Pr[\cV\to\top]\leq \frac{7}{8}+\negl(\secp).
\end{eqnarray}
Toward contradiction, suppose that there is a probabilistic polynomial-time cheating prover $\cA$ and a polynomial $\poly$ such that 
\begin{eqnarray*}
\Pr[\cV\to\top]\geq  \frac{7}{8}+\frac{1}{\poly(\secp)}
\end{eqnarray*}
for infinitely many $\secp$. In the following, we focus on such $\secp$. 
Let $\mathsf{ST}_\A$ be $\A$'s state (including the transcript and its own randomness) right after finishing the remote state preparation protocol run in Step \ref{item:RSP} and $\{x_0,x_1\}$ be $\mathcal{V}$'s output for the remote state preparation protocol. Then, by a standard averaging argument, for  $\frac{1}{2\poly(\secp)}$-fraction of $(\mathsf{ST}_\A,\{x_0,x_1\})$, we have 
\begin{eqnarray} \label{eq:contradiction}
\Pr[\cV\to\top~|~(\mathsf{ST}_\A,\{x_0,x_1\})]\geq  \frac{7}{8}+\frac{1}{2\poly(\secp)},
\end{eqnarray}
where 
$\Pr[\cV\to\top~|~(\mathsf{ST}_\A,\{x_0,x_1\})]$ denotes $\cV$'s acceptance probability conditioned on a fixed $(\mathsf{ST}_\A,\{x_0,x_1\})$.
We fix such $(\mathsf{ST}_\A,\{x_0,x_1\})$ until Equation \ref{eq:succ_E}.

We define the following probabilities all of which are conditioned on the fixed value of $(\mathsf{ST}_\A,\{x_0,x_1\})$:
\begin{description}
\item[$p_0$:] The probability that $\cV$ returns $\top$ conditioned on   $v_1=0$.
\item[$p_1$:] The probability that $\cV$ returns $\top$ conditioned on  $v_1=1$.
\item[$p_{1,0}$:] The probability that $\cV$ returns $\top$ conditioned on $v_1=1$ and $v_2=0$.
\item[$p_{1,1}$:] The probability that $\cV$ returns $\top$ conditioned on  $v_1=1$ and $v_2=1$.
\end{description}
Clearly, we have 
\begin{eqnarray} \label{eq:p0p1}
\Pr[\cV\to\top|(\mathsf{ST}_\A,\{x_0,x_1\})]= \frac{p_0+p_1}{2}
\end{eqnarray}
and 
\begin{eqnarray}\label{eq:p10p11}
p_{1}= \frac{p_{1,0}+p_{1,1}}{2}.
\end{eqnarray}
By Inequality \ref{eq:contradiction}, Equation \ref{eq:p0p1}, and a trivial inequality $p_0,p_1\le 1$, we have 
\begin{align}\label{eq:p0_lowerbound}
    p_0 \geq \frac{3}{4}+\frac{1}{\poly(\secp)}
\end{align}  
and 
\begin{align}\label{eq:p1_lowerbound}
    p_1 \geq \frac{3}{4}+\frac{1}{\poly(\secp)}.
\end{align}

Let $\B$ be a classical deterministic polynomial-time algorithm that works as follows:
\begin{enumerate}
    \item $\B$ takes $\mathsf{ST}_{\A}$ and $r\in \bit^n$ as input.
     \item $\B$ runs Step \ref{item:v1} of $\A$ whose state is initialized to $\mathsf{ST}_\A$ where $\B$ plays the role of $\cV$ with $v_1=1$ and the given $r$.  Let $d\in \bit^n$ be the message sent from $\A$ to $\cV$ and $\mathsf{ST}'_{\A}$ be $\A$'s state at this point. 
    \item $\B$ runs Step~\ref{item:v2} of $\A$ whose state is initialized to $\mathsf{ST}'_\A$ where  $\B$ plays the role of $\cV$ with $v_2=0$. Let $\eta_{1,0}$ be the message sent from $\A$ to $\cV$. 
    \item  $\B$ runs Step \ref{item:v2} of $\A$ whose state is initialized to $\mathsf{ST}'_\A$ where  $\B$ plays the role of $\cV$ withs $v_2=1$. Let $\eta_{1,1}$ be the message sent from $\A$ to $\cV$.  We note that this step is possible because $\A$ is classical and in particular $\mathsf{ST}'_\A$ is classical and thus can be copied. 
    \item Output  $\eta_{1,0}\oplus \eta_{1,1}$. 
\end{enumerate} 
By the union bound, the probability that both $(d,\eta_{1,0})$ and $(d,\eta_{1,1})$ pass the verification is at least 
\begin{align*}
    1-(1-p_{1,0})-(1-p_{1,1})=-1+2p_1\geq \frac{1}{2}+\frac{1}{\poly(\secp)},
\end{align*}
where the equation follows from Equation~\ref{eq:p10p11} and
the inequality follows from Inequality \ref{eq:p1_lowerbound}. 
When this occurs, for each $v_2\in \bit$, we have
\begin{eqnarray*}
    (r\cdot x_0=r\cdot x_1) \wedge
    (\eta_{1,v_2}=r\cdot x_0),
    \end{eqnarray*}
    or
    \begin{eqnarray*}
    (r\cdot x_0\neq r\cdot x_1) \wedge
    (\eta_{1,v_2}=v_2\oplus d\cdot(x_0\oplus x_1)).
    \end{eqnarray*}
(Remark that the same $d$ is used for both cases of $v_2=0$ and $v_2=1$.) 
This implies that
\begin{align*}
    \eta_{1,0}\oplus \eta_{1,1}=r\cdot(x_0\oplus x_1).
\end{align*}
Therefore, we have 
\begin{align*}
    \Pr_{r\gets \bit^n}[\B(\mathsf{ST}_{\A},r)=r\cdot(x_0\oplus x_1)]\geq \frac{1}{2}+\frac{1}{\poly(\secp)}.
\end{align*}
Thus, by the Goldreich-Levin theorem~\cite{GL89}, there is a probabilistic polynomial-time algorithm $\mathcal{E}$ such that 
\begin{align} \label{eq:succ_E}
    \Pr[\mathcal{E}(\mathsf{ST}_{\A})=x_0\oplus x_1]\geq\frac{1}{\poly'(\secp)}
\end{align}
for some polynomial $\poly'$. (Remark that what we showed so far is that the above hold for $\frac{1}{2\poly(\secp)}$-fraction of $(\mathsf{ST}_{\A},\{x_0,x_1\})$.)

Then, we construct a probabilistic polynomial-time algorithm $\mathcal{C}$ that breaks the security of the remote state preparation protocol as follows:
\begin{enumerate}
    \item $\mathcal{C}$ interacts with $\cV$ in the same way as $\A$ does in Step \ref{item:RSP} of the proof of quantumness protocol. Let $\mathsf{ST}_\A$ be $\A$'s state after completing this stage. Note that $\{x_0,x_1\}$ is implicitly defined as an outcome of $\cV$ for the remote state preparation protocol.
    \item $\mathcal{C}$ runs $\A$ for $v_1=0$ and $r\gets \bit^n$ to get the response $x'$. 
    \item $\mathcal{C}$ runs $\mathcal{E}(\mathsf{ST}_\A)$ to get the output $z$. 
    \item $\mathcal{C}$ outputs $\{x',x'\oplus z\}$. 
\end{enumerate}
 For $\frac{1}{2\poly(\secp)}$-fraction of $(\mathsf{ST}_\A,\{x_0,x_1\})$, 
by Inequalities \ref{eq:p0_lowerbound} and \ref{eq:succ_E}, 
we have 
\begin{align*}
    \Pr[x'\in \{x_0,x_1\}|(\mathsf{ST}_\A,\{x_0,x_1\})]\geq  \frac{3}{4}+\frac{1}{\poly(\secp)}
\end{align*}
and 
\begin{align*}
    \Pr[z=x_0\oplus x_1 |(\mathsf{ST}_\A,\{x_0,x_1\})]\geq \frac{1}{\poly'(\secp)}.
\end{align*}
Moreover, the two events $x'\in \{x_0,x_1\}$ and $z=x_0\oplus x_1$ are independent once we fix $(\mathsf{ST}_\A,\{x_0,x_1\})$. 
Therefore, for $\frac{1}{2\poly(\secp)}$-fraction of $(\mathsf{ST}_\A,\{x_0,x_1\})$, we have 
\begin{align*}
    \Pr[
    x'\in \{x_0,x_1\} \wedge
    z=x_0\oplus x_1 |(\mathsf{ST}_\A,\{x_0,x_1\})]\geq \frac{3}{4\poly'(\secp)}.
\end{align*}
Therefore, we have 
\begin{align*}
    \Pr[\mathcal{C}\rightarrow \{x_0,x_1\}]\geq \frac{3}{8\poly(\secp)\poly'(\secp)}.
\end{align*}
This contradicts the security of the remote state preparation protocol (\cref{def:RSP}). Therefore, Equation~\ref{eq:goal} holds and the proof of soundness is completed. 

\setcounter{tocdepth}{2}
\tableofcontents
\end{document}